\documentclass[pdflatex,sn-mathphys-num]{sn-jnl}
\usepackage{graphicx}%
\usepackage{multirow}%
\usepackage{mathrsfs}%
\usepackage{amsmath, amssymb, amsfonts}
\usepackage{enumitem}

\usepackage[title]{appendix}%
\usepackage{xcolor}%
\usepackage{textcomp}%
\usepackage{manyfoot}%
\usepackage{pgfplots}
\usepackage{booktabs}%
\usepackage{algorithm}%
\usepackage{subcaption}

\usepackage{algpseudocode}%
\usepackage{listings}%
\usepackage{tikz}

\usetikzlibrary{arrows.meta, positioning, shapes.geometric}

\theoremstyle{thmstyleone}%
\newtheorem{theorem}{Theorem}

\newtheorem{lemma}[theorem]{Lemma}%

\theoremstyle{thmstyletwo}%

\theoremstyle{thmstylethree}%

\begin{document}

\title[Article Title]{New Algorithm for Structured OFDM Channel Estimation using Subgroup Duality
}

\author*[1]{\fnm{Demerson N.} \sur{Gonçalves}}\email{demerson.goncalves@cefet-rj}

\author[2]{\fnm{João T.} \sur{Dias}}\email{joao.dias@cefet-rj.br}

\affil*[1]{\orgdiv{Department of Mathematics}, \orgname{CEFET-RJ}, \orgaddress{ \city{Petrópolis}, \postcode{25620-003}, \state{RJ}, \country{Brazil}}}

\affil[2]{\orgdiv{Department of Telecommunication Engineering}, \orgname{CEFET-RJ}, \orgaddress{\city{Rio de Janeiro}, \state{RJ}, \country{Brazil}}}

\abstract{

This paper presents a group-theoretic framework for structured channel estimation in Orthogonal Frequency Division Multiplexing (OFDM).  
By modeling subcarriers as the cyclic group \(\mathbb{Z}_N\), we show that nulling a subgroup \(H \subseteq \mathbb{Z}_N\) constrains the channel impulse response to its annihilator \(H^\perp\) in the dual domain.  
A low-complexity estimator is proposed that detects such structure by evaluating energy concentration across candidate annihilators.  
Simulations demonstrate consistent gains in mean squared error, bit error rate, and throughput compared with least-squares and linear minimum mean square error baselines, achieving competitive performance with substantially lower complexity and preserved interpretability.
}


\maketitle

\section{Introduction}
\label{sec:introduction}

OFDM is the cornerstone of modern wireless communication systems, from 4G LTE to 5G NR, due to its robustness against frequency-selective fading and its efficient implementation via the Fast Fourier Transform (FFT) \cite{Coleri2002, LiStuber2006, Goldsmith_2005}. By dividing a wideband channel into multiple parallel narrowband subchannels, OFDM mitigates inter-symbol interference (ISI) with the aid of a simple cyclic prefix (CP), enabling high-data-rate transmission with manageable receiver complexity.

Nevertheless, conventional OFDM faces significant challenges in high-mobility scenarios and emerging next-generation use cases. In vehicular communications or high-speed rail, for example, the channel becomes doubly-selective, varying rapidly in both time and frequency \cite{Ai2014}. This time-frequency selectivity introduces inter-carrier interference (ICI), compromises the performance of the standard one-tap equalizer, and reduces the reliability of pilot-aided channel estimation as the channel coherence time shrinks \cite{Aoudia2021}. Mitigation strategies often resort to increased pilot density or extended CP, which directly reduce spectral efficiency, a scarce resource in uplink scenarios and particularly critical for envisioned 6G applications requiring extreme mobility and data rates \cite{Chafii2023, Tong2022}.

Several research directions have been pursued to address these limitations. One line seeks to replace OFDM with more Doppler-resilient waveforms, such as Orthogonal Time Frequency Space (OTFS) modulation \cite{Hadani2017}. Another, closer in spirit to our work, explores the \emph{neural augmentation} of OFDM. Deep learning-based receivers, such as Deep-OFDM, learn to perform end-to-end modulation and equalization in pilot-sparse, high-Doppler regimes, achieving impressive performance gains \cite{Ye2017, Honkala2021, Aoudia2021}. While effective, these approaches often act as black boxes, requiring substantial data and computation and offering limited interpretability.

Classical approaches to channel estimation remain foundational. Pilot-aided least squares (LS) and linear minimum mean square error (LMMSE) estimators, along with interpolation strategies, represent the baseline in OFDM receivers \cite{Coleri2002,Barhumi2003,LiStuber2006}. More advanced refinements employ decision-directed updates, time–frequency filtering, or Kalman prediction to track variations \cite{Cho2010}. In parallel, compressed-sensing techniques exploit channel sparsity in the delay–Doppler domain to reduce estimation overhead \cite{Bajwa2010}. These models highlight that practical multipath channels often admit sparse representations, though most works focus on statistical sparsity rather than deterministic structure.

Beyond sparsity, other forms of structural regularity arise naturally. Uniform linear arrays (ULAs) induce spatial periodicity that interacts with OFDM processing, leading to predictable spectral patterns \cite{VanTrees2002, HeathLozano2018}. Reconfigurable Intelligent Surfaces (RIS) create programmable reflection profiles where linear or periodic phase shifts yield comb-like spectral responses with selective nulling across subcarriers \cite{Basar2019, DiRenzo2020}. Such examples indicate that wireless channels frequently exhibit \emph{regular} rather than arbitrary fading patterns.

In this work, we argue that these regularities can be captured through an algebraic perspective. The mathematical foundation builds on viewing the $N$ subcarriers in an OFDM symbol as elements of the cyclic group $\mathbb{Z}_N$, where subcarrier indices operate modulo $N$. Within this framework, a \emph{subgroup} $H \subseteq \mathbb{Z}_N$ corresponds to a uniformly spaced subset of subcarriers that exhibits closure under modular addition. The corresponding \emph{annihilator} $H^\perp \subseteq \widehat{\mathbb{Z}}_N$ represents the set of time-domain delays where the channel impulse response can be non-zero. This subgroup–annihilator duality emerges naturally from Fourier analysis and provides a principled way to model channels that systematically null specific, regularly spaced subcarrier patterns.

We observe that many practical scenarios, including channels influenced by ULAs, RIS, or periodic scattering environments, exhibit such structured nulling. When the frequency response nulls a subgroup $H \subseteq \mathbb{Z}_N$ of subcarriers, the time-domain impulse response becomes constrained to $H^\perp \subseteq \widehat{\mathbb{Z}}_N$. This constraint reduces the effective channel dimension, enabling more efficient estimation with fewer pilots. Our approach leverages this algebraic structure to develop interpretable and computationally efficient channel estimators, offering a mathematically grounded alternative to purely data-driven methods.

Our contributions can be summarized in four main points:
\begin{enumerate}[label=(\roman*)]
    \item We introduce a group-theoretic channel model for OFDM, in which the frequency response nulls a subgroup $H \subseteq \mathbb{Z}_N$, yielding a time-domain impulse response supported on its annihilator $H^\perp$.
    \item We propose a low-complexity estimator that evaluates energy concentration of an IDFT-based estimate over candidate annihilators and applies a threshold criterion for validation.
    \item We connect the algebraic framework to practical engineering scenarios such as indoor corridors, ULAs, and RIS, demonstrating how subgroup structures naturally arise in propagation.
    \item We present simulation results showing consistent gains in MSE, BER, and throughput over standard baselines across diverse OFDM configurations and SNR regimes.
\end{enumerate}

The remainder of this paper is organized as follows. Section ~\ref{sec:physical_motivation} describes the physical motivation and modeling scenarios. Section ~\ref{sec:group_model}  presents the system model and group-theoretic formulation. Section ~\ref{sec:energy_algorithm} details the proposed estimation algorithm. Section \ref{sec:experiments} provides simulation results, and Section \ref{sec:conclusion} concludes the paper.

\section{Physical Motivation and Modeling Scenarios}
\label{sec:physical_motivation}

Conventional OFDM channel estimation typically assumes unstructured fading, treating the channel frequency response as an arbitrary complex vector \cite{Coleri2002,Barhumi2003}. However, numerous real-world propagation environments exhibit deterministic regularities, whether, geometric, architectural, or engineered, that induce systematic patterns in both the delay and frequency domains. These structural constraints align naturally with the subgroup–annihilator duality formalized in Section~\ref{sec:group_model}, providing a principled framework for exploiting such regularities. Below we detail representative scenarios where these algebraic structures emerge from physical propagation mechanisms, supported by both theoretical analysis and empirical observations.

\subsection{Symmetric Multipath in Indoor Channels}

Indoor environments with regular geometries, such as corridors, tunnels, or rectangular rooms, produce characteristic multipath patterns. Parallel reflecting surfaces generate echoes at equally spaced delays, creating an impulse response supported on a periodic lattice $\langle d \rangle \subseteq \mathbb{Z}_N$ \cite{Sayed2003, Molisch2012}. By Fourier duality, this time-domain periodicity manifests as systematic nulls every $d$ subcarriers in the frequency domain, corresponding precisely to the annihilator subgroup $H^\perp = \langle N/d\rangle$.

Experimental studies in waveguide-like environments \cite{Durgin2003} have measured such periodic nulling patterns, with the null spacing determined by the inverse of the dominant delay difference. For a corridor of width $W$, the dominant delay difference is $\Delta\tau = (2W\cos\theta)/c$, producing frequency nulls at intervals $\Delta f = 1/\Delta\tau$. When mapped to OFDM subcarriers, this yields the subgroup structure $H = \langle \lfloor N\Delta\tau/T_s \rfloor \rangle$, where $T_s$ is the OFDM  symbol duration.

\subsection{Uniform Linear Arrays}

Uniform Linear Arrays (ULAs) impose spatial periodicity that translates directly into algebraic channel structure through the array manifold. For a ULA with $M$ elements spaced by $d_a$, the array response vector for a plane wave arriving from angle $\theta$ is:
\begin{equation}
\mathbf{a}(\theta) = [1, e^{-j2\pi d_a\sin\theta/\lambda}, \dots, e^{-j2\pi(M-1)d_a\sin\theta/\lambda}]^T.
\end{equation}
When multiple paths exhibit angular symmetry or are aligned with the array grating lobes, the effective channel impulse response becomes supported on a periodic set $\{n_p\} \subseteq \langle d\rangle$ \cite{VanTrees2002,HeathLozano2018}.

After beamforming or spatial processing, the composite frequency response exhibits $N/d$-periodicity, reflecting redundancy across the annihilator $H_d^\perp$. This structure has been exploited in massive MIMO systems \cite{Rusek2013} where the inherent spatial oversampling creates predictable patterns in the equivalent single-input single-output (SISO) channel.

\subsection{RIS-Induced Comb-Like Spectra}

Reconfigurable Intelligent Surfaces (RIS) enable deliberate engineering of channel responses through programmable reflection profiles. By applying periodic phase gradients $\phi_m = 2\pi m/q + \phi_0$ across RIS elements, one can create spectral combs that selectively activate subcarriers in the subgroup $\langle q\rangle$ \cite{Basar2019,DiRenzo2020}. 

The composite channel, incorporating both direct and RIS-reflected paths, exhibits time-domain support restricted to $H_q^\perp = \langle N/q\rangle$. Recent experimental demonstrations \cite{Tang2021} have verified this comb-like spectral shaping, showing that RIS can indeed impose algebraic structure on wireless channels. This programmability transforms RIS from merely a beamforming tool into a channel structure engineering platform, opening new possibilities for structured estimation and equalization.

\subsection{Additional Structured Scenarios}

Beyond these primary cases, several other scenarios exhibit similar algebraic regularity:
\begin{enumerate}[label=(\roman*)]
\item \textbf{Velocity-induced periodicity}: In high-mobility scenarios with constant velocity, the Doppler spectrum becomes concentrated around specific frequencies, creating periodic time variations that translate into structured frequency-domain patterns \cite{Zhou2020}.

\item \textbf{Intelligent reflecting surfaces}: Beyond RIS, other electromagnetic skin technologies \cite{Wu2022} can impose geometric constraints on propagation paths, leading to structured sparsity in angular and delay domains.

\item \textbf{Acoustic and underwater channels}: In bounded media with strong waveguide effects, modal propagation creates inherent frequency-domain periodicity \cite{Stojanovic2008}.
\end{enumerate}

These diverse examples demonstrate that the subgroup structure proposed in this work is not merely an algebraic abstraction but physically realizable across multiple domains. Whether through natural geometric symmetries, engineered antenna arrays, or programmable metasurfaces, the resulting sparsity and periodicity patterns can be rigorously described using finite-group duality, providing a unified framework for exploiting structural regularities in channel estimation.

\section{Group-Theoretic Channel Model}
\label{sec:group_model}

This section develops a unified group-theoretic framework for modeling structured OFDM channels. 
While conventional OFDM systems are well understood, our goal is to reinterpret their frequency–time correspondence through the lens of algebraic duality, 
providing an explicit connection between periodic spectral patterns and sparse temporal supports. 
We begin by recalling the standard OFDM system model and then present its algebraic formulation.

\subsection{OFDM System Model}

Consider a standard OFDM system with \(N\) subcarriers transmitting complex symbols
\(\{X_k\}_{k=0}^{N-1}\).
The time-domain signal after inverse DFT (IDFT) is
\begin{equation}
x[n] = \frac{1}{\sqrt{N}} \sum_{k=0}^{N-1} X_k e^{j 2\pi k n / N},
\quad n = 0, 1, \ldots, N-1,
\label{eq:ofdm_tx_signal}
\end{equation}
where the normalization factor ensures a unitary DFT pair.  
The signal \(x[n]\) is transmitted through a linear time-invariant (LTI) channel
with impulse response \(h[n]\) of length \(L_h < N\),
and is corrupted by additive white Gaussian noise \(w[n]\).

After cyclic prefix (CP) insertion and removal,
the circular convolution between \(x[n]\) and \(h[n]\) is preserved provided that
\begin{equation}
N_{\mathrm{CP}} \ge L_h,
\label{eq:cp_condition}
\end{equation}
ensuring orthogonality among subcarriers and avoiding intersymbol interference (ISI).  
After DFT demodulation, the received symbols satisfy
\begin{equation}
Y[k] = H[k] X_k + W[k], 
\quad k = 0, 1, \ldots, N-1,
\label{eq:ofdm_freq_model}
\end{equation}
where
\(
H[k] = \sum_{n=0}^{L_h-1} h[n] e^{-j 2\pi k n / N}
\)
is the frequency response of the channel.

Figure~\ref{fig:ofdm-group-estimator} illustrates the complete system model.
The proposed \textit{Group-Based Channel Estimation} module (in green)
exploits the subgroup–annihilator structure of the channel.
A cyclic subgroup \(H = \{0, d, 2d, \dots\}\) of periodically null subcarriers
corresponds to a structured time-domain support
(\(\mathrm{supp}(h) \subseteq H^\perp = \langle N/d \rangle\)),
enabling efficient and interpretable channel recovery.

\begin{figure}[ht]
\centering
\begin{tikzpicture}[
    block/.style={rectangle, draw, fill=blue!10, rounded corners=1pt,
                  minimum width=2.0cm, minimum height=0.9cm, align=center},
    channel/.style={rectangle, draw, fill=red!10, rounded corners=1pt,
                    minimum width=2.1cm, minimum height=0.9cm, align=center},
    proposed/.style={rectangle, draw, fill=green!20, rounded corners=1pt,
                     minimum width=2.5cm, minimum height=1cm, align=center},
    arrow/.style={-Stealth, thick},
    node distance=0.9cm and 1.2cm
]

\node (input) {$b$};
\node[below=3.2cm of input] (output) {$\hat{b}$};

\node[block, right=of input] (map) {Mapping};
\node[block, right=of map] (idft) {IDFT};
\node[block, right=of idft] (addcp) {Add CP};
\node[block, right=of addcp] (dac) {DAC};

\node[channel, below=of dac] (channel) {Channel\\(+ Noise)};

\node[block, below=of channel] (adc) {ADC};
\node[block, left=of adc] (rmcp) {Remove CP};
\node[block, left=of rmcp] (dft) {DFT +\\Equalizer};
\node[block, left=of dft] (demap) {Demapping};

\node[proposed, below=0.8cm of rmcp] (est) {Group-Based\\Channel Estimation};

\draw[arrow] (input) -- (map);
\draw[arrow] (map) -- node[above] {$s$} (idft);
\draw[arrow] (idft) -- node[above] {$x[n]$} (addcp);
\draw[arrow] (addcp) -- (dac);
\draw[arrow] (dac) -- (channel);
\draw[arrow] (channel) -- (adc);
\draw[arrow] (adc) -- (rmcp);
\draw[arrow] (rmcp) -- node[above] {$y[n]$} (dft);
\draw[arrow] (dft) -- node[above] {$\tilde{s}$} (demap);
\draw[arrow] (demap) -- (output);

\draw[arrow] (rmcp.south) -- (est.north);
\draw[arrow] (est.west) -| (dft.south);

\node[above=0.2cm of map, font=\small\bfseries, text=blue!70!black] {Transmitter};
\node[below=0.2cm of demap, font=\small\bfseries, text=blue!70!black] {Receiver};

\node[above=1.8cm of dft, font=\scriptsize\itshape, text=blue!60!black, align=center]
{Frequency domain: nulls in $H = \{0,d,2d,\dots\}$};

\node[below=0.1cm of est, font=\scriptsize\itshape, text=green!50!black, align=center]
{Time-domain support: $\mathrm{supp}(h)\subseteq H^\perp = \langle N/d\rangle$};

\node[below left=3cm and -0.7cm of adc, font=\scriptsize, align=left] {
\textcolor{blue!70!black}{\colorbox{blue!10}{\rule{0.15cm}{0.cm}}} OFDM blocks\quad
\textcolor{red!70!black}{\colorbox{red!15}{\rule{0.15cm}{0.cm}}} Channel + Noise\quad
\textcolor{green!60!black}{\colorbox{green!25}{\rule{0.15cm}{0.cm}}} Proposed estimator
};
\end{tikzpicture}
\caption{Block diagram of the OFDM system with group-based channel estimation.}
\label{fig:ofdm-group-estimator}
\end{figure}
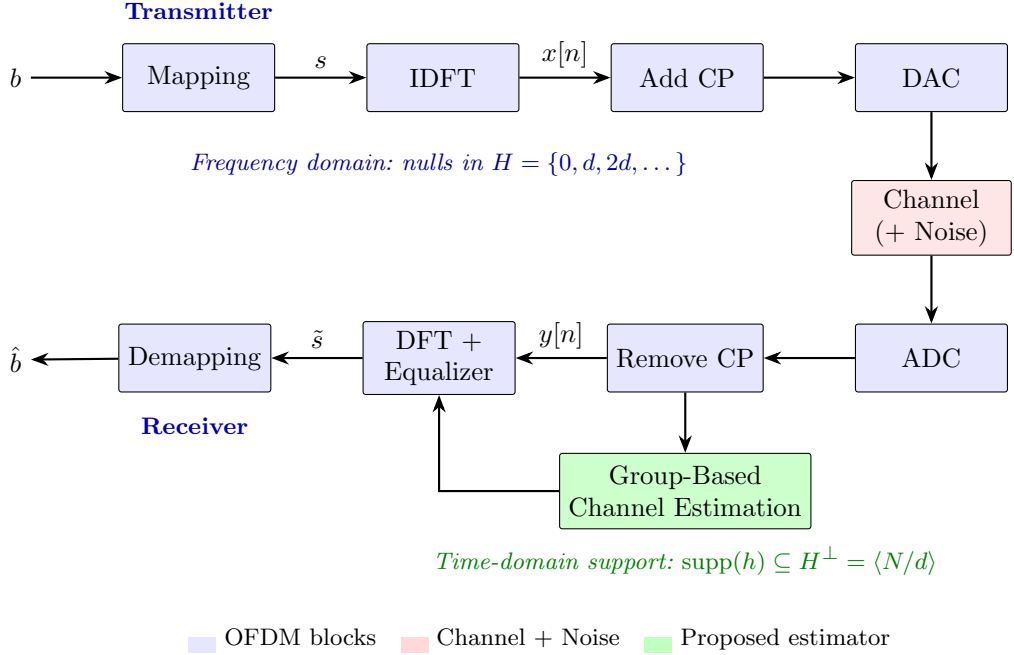

\subsection{Group-Theoretic Representation of OFDM}

To formalize the structure observed in OFDM channels, we recall a few basic definitions.  
A \emph{group} \( (G,+) \) is a set endowed with an associative binary operation, possessing an identity element and inverses.  
A \emph{subgroup} \( H \le G \) is a subset that is itself a group under the same operation.  
If the group operation is commutative, that is, \(a+b=b+a\) for all \(a,b\in G\), the group is said to be \emph{abelian}.  
A group is \emph{cyclic} if there exists an element \(g\in G\) such that every element of \(G\) can be written as a multiple (or power) of \(g\); we write \(G=\langle g\rangle\) and call \(g\) a \emph{generator}.  
A \emph{homomorphism} is a structure-preserving map between groups, i.e., 
\(\phi(a+b)=\phi(a)\phi(b)\).  
A \emph{character} of a finite abelian group is a homomorphism from \(G\) to the multiplicative unit circle in \(\mathbb{C}\).  
For the cyclic group \( \mathbb{Z}_N = \langle 1 \rangle \), the characters take the exponential form
\begin{equation}
\chi_k(n)=e^{2\pi i k n/N},
\label{eq:character}
\end{equation}
which directly corresponds to the DFT basis used in OFDM.

An OFDM system partitions the available bandwidth into \(N\) orthogonal subcarriers indexed by \(\{0, 1, \ldots, N-1\}\).  
These indices can be naturally modeled by the additive cyclic group \(\mathbb{Z}_N\), which captures the modular arithmetic inherent to subcarrier indexing.  
A \emph{subgroup} \(H \subseteq \mathbb{Z}_N\) represents a periodic subset of subcarriers, and hence provides a compact description of regular spectral structures.

\subsection{Subgroup–Annihilator Duality}
\label{sec:annihilator_duality}

Let \(G=\mathbb{Z}_N\) denote the additive cyclic group and \(\widehat{G}\) its dual, composed of characters \(\chi:G \to \mathbb{C}\) as defined in Eq.~(\ref{eq:character}).  
For any subgroup \(H\le G\), its \emph{annihilator} is defined as
\begin{equation}
H^\perp \triangleq \{\,\chi\in\widehat{G}:\chi(h)=1,\ \forall\,h\in H\,\},
\label{eq:annihilator_char_def}
\end{equation}
namely, the set of characters that act trivially on \(H\).

In the finite cyclic case, the dual group \(\widehat{G}\) is canonically isomorphic to \(\mathbb{Z}_N\) via the character family of Eq.~(\ref{eq:character}),
under which the annihilator takes the index form
\begin{equation}
H^\perp \cong \{\,k\in\mathbb{Z}_N:\ e^{\frac{2\pi i}{N}kh}=1,\ \forall\,h\in H\,\}.
\label{eq:annihilator_index_form}
\end{equation}
If \(H=\langle d\rangle=\{0,d,2d,\dots\}\) with \(d\mid N\), then
\begin{equation}
H^\perp \cong \langle N/d\rangle=\{0,\,N/d,\,2N/d,\,\dots\},
\qquad |H|=\frac{N}{d},\ \ |H^\perp|=d.
\label{eq:annihilator_cyclic_case}
\end{equation}
Pontryagin duality guarantees the bidual identity
\begin{equation}
(H^\perp)^\perp = H,
\label{eq:bidual_identity}
\end{equation}
expressing the perfect algebraic symmetry between a subgroup and its annihilator.

In the OFDM setting, identifying \(\widehat{G}\cong\mathbb{Z}_N\) allows characters to be indexed by frequency bins \(k\).  
A subgroup \(H=\langle d\rangle\) representing periodically null (or suppressed) frequency components corresponds, through Fourier duality, to a time-domain channel whose nonzero coefficients are confined to \(H^\perp=\langle N/d\rangle\).  
Conversely, a short time-domain support induces periodic spectral patterns with zeros across frequency cosets of \(\langle d\rangle\).  
This duality provides the theoretical foundation for the proposed estimator, which searches nested annihilator subspaces to capture at least a \((1-\varepsilon)\) fraction of the total channel energy using the smallest possible support. Figure~\ref{fig:duality_illustration} illustrates this correspondence for \(N=12\) and \(d=3\),
where periodic nulls in frequency (\(H=\{0,3,6,9\}\))
map to a sparse time-domain support confined to
\(H^\perp=\{0,4,8\}\).


\begin{figure}[t]
\centering
\begin{tikzpicture}[scale=0.95, every node/.style={font=\small}]

\draw[->] (-0.3,0) -- (6.8,0) node[right] {Frequency indices $k \in \mathbb{Z}_{12}$};

\foreach \x in {0,0.55,...,6.05}
  \filldraw[gray!60] (\x,0) circle (1.3pt);

\foreach \x in {0,1.65,3.3,4.95}
  \filldraw[blue!70!black] (\x,0) circle (2.3pt);

\node[below, blue!70!black] at (3.2,-0.5)
  {$H = \langle 3 \rangle = \{0,3,6,9\}$  (periodic spectral nulls)};

\draw[->] (-0.3,-3) -- (6.8,-3) node[right] {Time indices $n \in \mathbb{Z}_{12}$};

\foreach \x in {0,0.55,...,6.05}
  \filldraw[gray!60] (\x,-3) circle (1.3pt);

\foreach \x in {0,2.2,4.4}
  \filldraw[red!80!black] (\x,-3) circle (2.3pt);

\node[below, red!80!black] at (3.2,-3.55)
  {$H^{\perp} = \langle 4 \rangle = \{0,4,8\}$  (sparse time-domain support)};

\draw[->, thick, dashed, gray!60!black] (3.2,-0.3) -- (3.2,-2.7)
  node[midway, right, black!70] {$\mathcal{F}^{-1}$};
\draw[->, thick, dashed, gray!60!black] (4.6,-2.7) -- (4.6,-0.3)
  node[midway, right, black!70] {$\mathcal{F}$};

\node[below right, font=\scriptsize, align=left, black!60] at (-0.3,-4.0) {
\textcolor{blue!70!black}{\rule{0.3cm}{0.1cm}} Frequency-domain subgroup $H$ \quad
\textcolor{red!80!black}{\rule{0.3cm}{0.1cm}} Time-domain annihilator $H^\perp$
};

\end{tikzpicture}
\caption{Subgroup–annihilator duality in OFDM for $N=12$ and $d=3$.
}
\label{fig:duality_illustration}
\end{figure}
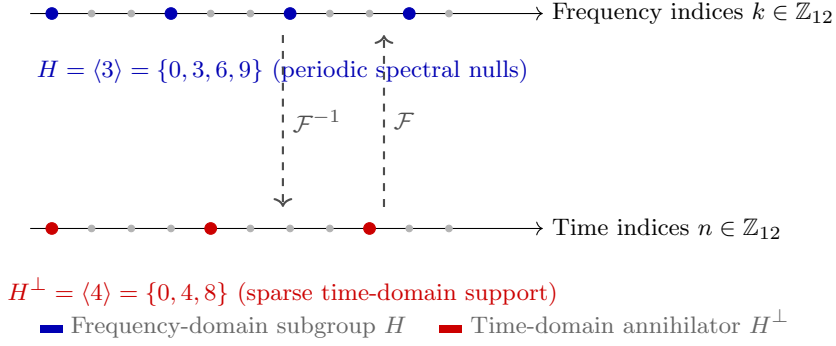

\section{Energy-Constrained Subgroup Estimation Algorithm}
\label{sec:energy_algorithm}

The proposed estimation algorithm exploits the algebraic structure of the channel to perform structured and energy-aware recovery. 
The key principle is that most of the channel energy is concentrated within a small annihilator subspace associated with a subgroup of \(\mathbb{Z}_N\). 
We first formalize this concentration property and then translate it into a constructive estimation rule.

Given the channel impulse response \(h = (h_0,\dots,h_{L-1}) \in \mathbb{C}^L\)
and a subgroup \(H^\perp \subseteq \mathbb{Z}_N\) representing its assumed time-domain support,
define the partial and total energy as
\begin{equation}
E_{H^\perp} = \sum_{n\in H^\perp} |h_n|^2,
\qquad
E_{\mathrm{total}} = \sum_{n=0}^{N-1} |h_n|^2.
\label{eq:energy_concentration_def}
\end{equation}

\begin{lemma}[Energy Concentration]\label{energy_concentration_lemma}
Let \(h\in\mathbb{C}^L\) be a channel whose time-domain energy is approximately sparse.
For any tolerance parameter \(\varepsilon \in (0,1)\),
there exists a minimal subgroup \(H_\varepsilon^\perp \subseteq \mathbb{Z}_N\)
such that
\begin{equation}
E_{H_\varepsilon^\perp} \ge (1-\varepsilon)\,E_{\mathrm{total}}.
\label{eq:energy_concentration}
\end{equation}
Moreover, \(H_\varepsilon^\perp\) corresponds to the smallest annihilator subspace
that preserves at least a fraction \((1-\varepsilon)\) of the total channel energy.
\end{lemma}

\begin{proof}
Let \(\mathcal{U}_d \subset \mathbb{C}^N\) denote the subspace of vectors
whose nonzero components are confined to the annihilator subgroup
\(H_d^\perp = \{\,0, N/d, 2N/d, \dots, (d-1)N/d\,\}\).
Let \(P_d : \mathbb{C}^N \to \mathcal{U}_d\) be the orthogonal projector onto that subspace,
so that
\[
(P_d h)[n] =
\begin{cases}
h[n], & n \in H_d^\perp,\\[2pt]
0, & \text{otherwise.}
\end{cases}
\]
The projection error satisfies
\begin{equation}
\|h - P_d h\|_2^2 = \|h\|_2^2 - \|P_d h\|_2^2,
\label{eq:projection_error}
\end{equation}
where \(\|P_d h\|_2^2 = E_{H_d^\perp}\) and \(\|h\|_2^2 = E_{\mathrm{total}}\).
Hence,
\[
\frac{E_{H_d^\perp}}{E_{\mathrm{total}}}
= 1 - \frac{\|h - P_d h\|_2^2}{\|h\|_2^2}.
\]
Imposing \(E_{H_d^\perp} / E_{\mathrm{total}} \ge 1 - \varepsilon\)
is, therefore, equivalent to bounding the relative projection error as
\begin{equation}
\frac{\|h - P_d h\|_2^2}{\|h\|_2^2} \le \varepsilon.
\label{eq:relative_error_bound}
\end{equation}
Since the subspaces \(\mathcal{U}_d\) form a nested sequence
(\(\mathcal{U}_{d_1} \subseteq \mathcal{U}_{d_2}\) whenever \(d_1 \mid d_2\)),
the mapping \(d \mapsto \|P_d h\|_2^2\) is monotone nondecreasing.
Thus, there exists a smallest \(d^\ast\) such that
\[
\|P_{d^\ast} h\|_2^2 \ge (1-\varepsilon)\|h\|_2^2.
\]
Setting \(H_\varepsilon^\perp = H_{d^\ast}^\perp\) yields the desired subgroup
that captures at least \((1-\varepsilon)\) of the total energy
with minimal support, completing the proof.
\end{proof}

The above lemma provides the theoretical foundation for constructing an efficient search procedure. 
Let the pilot received in the frequency-domain be \(Y[k_p]\) and the transmitted symbols \(X[k_p]\). 
The least-squares frequency response estimate is
\[
\hat{H}[k_p] = \frac{Y[k_p]}{X[k_p]},
\]
and its time-domain counterpart is obtained via the inverse DFT:
\[
\hat{h}[n] = \mathrm{IDFT}\{\hat{H}[k]\}.
\]
For each divisor \(d \mid N\), define the candidate annihilator subgroup \(H_d^\perp = \langle N/d \rangle\)
and compute the normalized energy ratio
\begin{equation}
R_d = \frac{\sum_{n \in H_d^\perp} |\hat{h}[n]|^2}{\sum_{n=0}^{N-1} |\hat{h}[n]|^2}.
\label{eq:energy_ratio}
\end{equation}
The optimal subgroup corresponds to the smallest \(d\) satisfying \(R_d > 1 - \varepsilon\),
as guaranteed by Lemma~\ref{energy_concentration_lemma}. 
The complete estimation procedure is summarized in Algorithm~\ref{alg:channel-estimation}.

\begin{algorithm}[t]
\caption{Energy-Constrained Subgroup Channel Estimation}
\label{alg:channel-estimation}
\begin{algorithmic}[1]
\Require Received symbols $\{Y[k]\}$, transmitted symbols $\{X[k]\}$, number of subcarriers $N$, threshold $\varepsilon$
\Ensure Estimated annihilator $H_{\text{best}}$, sparse impulse response $\hat{h}_{\text{sparse}}[n]$
\State $\hat{H}[k] \gets Y[k]/X[k]$
\State $\hat{h}[n] \gets \mathrm{IDFT}\{\hat{H}[k]\}$
\State $E_{\mathrm{total}} \gets \sum_{n=0}^{N-1} |\hat{h}[n]|^2$
\For{each divisor $d \mid N$ (in increasing order)}
    \State $H_d^\perp \gets \{\, n \in \mathbb{Z}_N : n \bmod (N/d) = 0 \,\}$
    \State $R_d \gets \frac{\sum_{n \in H_d^\perp} |\hat{h}[n]|^2}{E_{\mathrm{total}}}$
    \If{$R_d > (1 - \varepsilon)$}
        \State $H_{\text{best}} \gets H_d^\perp$
        \State \textbf{break}
    \EndIf
\EndFor
\If{$H_{\text{best}}$ undefined}
    \State $\hat{h}_{\text{sparse}}[n] \gets 0$
\Else
    \State $\hat{h}_{\text{sparse}}[n] \gets
           \begin{cases}
             \hat{h}[n], & n \in H_{\text{best}},\\
             0, & \text{otherwise.}
           \end{cases}$
\EndIf
\State \Return $H_{\text{best}}, \hat{h}_{\text{sparse}}[n]$
\end{algorithmic}
\end{algorithm}

The criterion used in Algorithm~\ref{alg:channel-estimation} 
not only identifies the most compact annihilator subspace that retains a prescribed fraction of the total channel energy,
but also implicitly minimizes the mean-squared error (MSE) of the structured estimate.
This analytical connection between energy preservation and estimation accuracy
is made explicit in the following subsection.




\subsection{Energy–MSE Relation}
\label{sec:energy_mse}

The relationship between energy concentration in annihilator subspaces and the mean-squared error (MSE) of structured channel estimates provides the theoretical foundation of the proposed framework. 
This section establishes the fundamental energy–MSE duality that justifies the subgroup selection rule in Algorithm~\ref{alg:channel-estimation} and extends it to practical scenarios involving noisy least-squares (LS) estimation.

Consider a channel impulse response $h \in \mathbb{C}^L$ and its unconstrained LS estimate $\hat{h}$. 
Let $\hat{h}_{H^\perp_\varepsilon}$ denote the orthogonal projection of $\hat{h}$ onto the selected annihilator subspace $H^\perp_\varepsilon$. 
The total and partial energies are given by $E_{\mathrm{total}}$ and $E_{H^\perp_\varepsilon}$ as defined in Eq.~\eqref{eq:energy_concentration_def}.  

Under ideal conditions ($\hat{h} = h$), the orthogonality of the projection implies that the residual error energy equals the portion of $h$ lying outside $H^\perp_\varepsilon$. 
Consequently, the MSE of the structured estimate satisfies
\begin{equation}
\mathrm{MSE}(H^\perp_\varepsilon)
= \frac{1}{N}\big(E_{\mathrm{total}} - E_{H^\perp_\varepsilon}\big).
\label{eq:mse_energy_relation}
\end{equation}
Applying Lemma~\ref{energy_concentration_lemma} yields a worst-case MSE guarantee: 
if $E_{H^\perp_\varepsilon} \ge (1-\varepsilon)E_{\mathrm{total}}$, then
\begin{equation}
\mathrm{MSE}(H^\perp_\varepsilon)
\le
\frac{\varepsilon}{N}E_{\mathrm{total}}.
\label{eq:mse_bound}
\end{equation}
Hence, maximizing captured energy within $H^\perp_\varepsilon$ is equivalent to minimizing the MSE over the family of annihilator-constrained estimates. 
This establishes the energy–MSE duality that underlies Algorithm~\ref{alg:channel-estimation}.

In realistic conditions, the LS estimate is corrupted by noise: $\hat{h} = h + w$, where $w \sim \mathcal{CN}(0, \sigma^2 I)$ represents complex Gaussian estimation noise. 
The projected estimate is then
\[
\hat{h}_{H^\perp_\varepsilon} = P_{H^\perp_\varepsilon}\hat{h}
= P_{H^\perp_\varepsilon}h + P_{H^\perp_\varepsilon}w,
\]
where $P_{H^\perp_\varepsilon}$ denotes the orthogonal projector onto $H^\perp_\varepsilon$. 
Taking expectations over the noise yields
\begin{equation}
\mathbb{E}\big[\|h - \hat{h}_{H^\perp_\varepsilon}\|_2^2\big]
= (E_{\mathrm{total}} - E_{H^\perp_\varepsilon}) + \sigma^2 d^\ast,
\label{eq:mse_noise_decomp}
\end{equation}
where $d^\ast = |H^\perp_\varepsilon|$ is the subspace dimension. 
Normalizing by $N$, the expected MSE becomes
\begin{equation}
\mathrm{MSE}(H^\perp_\varepsilon)
= \frac{1}{N}\big(E_{\mathrm{total}} - E_{H^\perp_\varepsilon}\big)
+ \frac{\sigma^2 d^\ast}{N}.
\label{eq:mse_noisy}
\end{equation}

Equation~\eqref{eq:mse_noisy} highlights the trade-off between energy capture and noise amplification: 
larger subspaces increase $E_{H^\perp_\varepsilon}$ but also the additive term proportional to $d^\ast$. 
The optimal annihilator $H^\perp_\varepsilon$ therefore maximizes
\[
E_{H^\perp_\varepsilon} - \sigma^2 d^\ast,
\]
which is equivalent to minimizing the MSE. 
Algorithm~\ref{alg:channel-estimation} effectively implements this criterion by selecting the smallest annihilator that captures at least a $(1-\varepsilon)$ fraction of the total energy, thereby ensuring both theoretical optimality and robustness in noisy environments.

The established energy–MSE relationship thus provides a unified theoretical and practical justification for the proposed estimator. 
In ideal conditions, it guarantees exact optimality; under noise, it ensures near-optimal performance with graceful degradation controlled by $\sigma^2$ and the chosen threshold $\varepsilon$.

\subsection{Computational Complexity}

The computational complexity of channel estimation methods is a critical practical consideration for real-time OFDM systems. In the proposed group-based approach, the number of distinct subgroups of \(\mathbb{Z}_N\), which directly corresponds to the number of candidate annihilators \(H^\perp\), is governed by the divisor function \(\tau(N)\), the number of positive divisors of \(N\). This function exhibits unique growth properties: it attains its minimum \(\tau(N) = 2\) when \(N\) is prime, and for composite \(N\), its growth is subpolynomial, bounded by \(\tau(N) = O(N^\epsilon)\) for any \(\epsilon > 0\). Crucially, the average behavior of \(\tau(N)\) is far more restrained, satisfying the asymptotic relation \(\mathbb{E}_{N \leq x}[\tau(N)] \sim \log x\), which implies that the number of subgroup candidates grows logarithmically as \(N\) scales.

Each candidate subgroup \(H_d\) is uniquely defined by a divisor \(d \mid N\), and its annihilator \(H_d^\perp\) serves as the search domain for energy concentration. Since each evaluation requires \(O(N)\) operations, the total computational complexity of the estimation process is \(O(\tau(N) \cdot N)\). This advantage arises from the algebraic organization of the problem: instead of examining arbitrary subsets of indices, the estimation is restricted to annihilator subgroups arising from the divisors of \(N\), restricting the hypothesis space to a structured family of candidates grounded in group theory.

This complexity profile compares favorably with established benchmarks. The classical least squares (LS) estimator achieves \(O(N \log N)\) complexity through efficient FFT-based interpolation, while the minimum mean square error (MMSE) estimator typically requires \(O(N^3)\) operations for solving linear systems with channel covariance matrices.
The Table \ref{tab:complexity} summarizes the complexity of the tested estimators.

\begin{table}[ht]
\centering
\caption{Computational complexity comparison of channel estimation methods}
\label{tab:complexity}
\begin{tabular}{l|l|l}
\hline
\textbf{Method} & \textbf{Complexity} & \textbf{Characteristics} \\
\hline
LS & $O(N \log N)$ & Efficient but noise-sensitive \\
LMMSE & $O(N^3)$ & Optimal but computationally intensive \\
Group-Based & $O(\tau(N) \cdot N)$ & Structure-exploiting, near-linear \\
\hline
\end{tabular}
\end{table}

\section{Numerical Experiments}
\label{sec:experiments}

We evaluate the proposed subgroup–based channel estimator in a controlled OFDM simulation
designed to highlight how algebraic structure influences estimation accuracy and link performance. Two scenarios with different channel characteristics were chosen for algorithm performance analysis: a tapped-delay-line (TDL) channel model for tunnels~\cite{9305216} and the ITU Channel Model for Indoor Office~\cite{ITU-R:M.1225}. Initially, to analyze an idealized scenario, the TDL channel model was configured with spacing coinciding with $H^\perp$ and average power with exponential decay $(e^{-03\tau_i})$, the ITU Channel Model has spacing not coincident with $H^\perp$, which shows the worst-case scenario, it has a relative delay $[0~ 100~ 200~ 300~ 500~ 700] (ns)$ and average power $[0~ -3.6~ -7.2~ -10.8~ -18.0~ -25.2] (dB)$.

The system employs \( N = 256 \) subcarriers and QPSK modulation,
with the signal–to–noise ratio sampled on the grid
\(\text{SNR}_{\mathrm{dB}} \in \{0,5,10,\ldots,25\}\).
Structural variability is introduced through the generator set
\(\mathcal{D} = \{2,8,16,64,128\}\),
where each \(d\) defines a cyclic subgroup \(H = \langle d \rangle \subseteq \mathbb{Z}_N\)
and its annihilator \(H^\perp = \langle N/d \rangle\).
This pair induces a specific sparsity pattern across subcarriers,
governing the fraction of active tones \(\eta_d = 1 - 1/d\),
which also impacts the achievable throughput.
The complete simulation configuration is summarized in Table~\ref{tab:sim_config}.

\begin{table}[ht!]
\centering
\caption{Simulation parameters used in the numerical experiments.}
\label{tab:sim_config}
\begin{tabular}{l c l}
\toprule
\textbf{Parameter} & \textbf{Symbol} & \textbf{Value / Description} \\ \midrule
Number of OFDM subcarriers & $N$ & 256 \\
OFDM symbol duration & – & $12.8 \mu s$ \\
Cyclic prefix duration & $n_{\mathrm{CP}}$ & $N/8 = 32$ samples ($12.5\%$ of symbol) \\
Modulation schemes & – & QPSK \\
Annihilator subgroup sizes & $d$ & \{2, 8, 16, 64, 128\} \\
Threshold parameter & $\varepsilon$ & 0.15 \\
SNR range & – & 0–25 dB (step = 5 dB) \\
Monte Carlo iterations & $K(d)$ & 100–300 (adaptive to sparsity) \\
Estimators compared & – & LS, LMMSE, Subgroup-based \\
Performance metrics & – & MSE, SER, Effective Throughput \\
\bottomrule
\end{tabular}
\end{table}

\paragraph{Baselines and protocol.}
Unless otherwise stated, all constants and hyperparameters follow
Table~\ref{tab:sim_config}.
For each configuration \((d,\text{SNR})\),
we perform \(K(d)\) independent Monte Carlo trials.
Each trial generates a channel, injects complex Gaussian noise at the target SNR,
and transmits QPSK symbols over all active tones.
Frequency–domain equalization uses a channel estimate \(\widehat H_m\).
We compare three estimators:
(i)~LS, given by \(\widehat H_{\mathrm{LS}}[k] = Y[k]/X[k]\) under AWGN~\cite{PapoulisPillai2002}; (ii)~LMMSE with the channel correlation matrix and noise variance known to the receiver~\cite{701321}
and (iii)~the proposed Subgroup–based estimator, enforcing sparsity according to \(H^\perp\).
All models share identical preprocessing and hyperparameters (\(\varepsilon = 0.15\)
) to isolate the effects of structure and SNR.

\paragraph{Performance.}
We first analyzed the performance of the algorithms in the TDL channel model.
Fig.~\ref{fig:TDL} shows that the subgroup-based estimator consistently outperforms the LS methods for all SNR range in the tested scenario. It has an MSE close to the LMMSE method, and similar SER and throughput, for SNR above 10 dB.  This test shows that, in the SNR range of interest, which is where communication systems actually operate (above 10 dB), the proposed method performs similar to the optimal method.


\begin{figure}[h!] 
  \centering 
  \begin{subfigure}[b]{0.30\textwidth} 
    \centering
    \includegraphics[width=\linewidth]{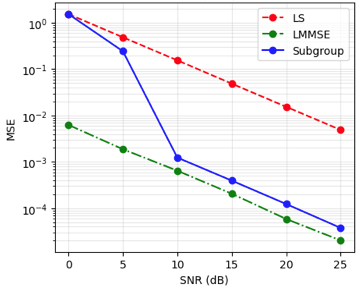}
    \caption{MSE comparison for the estimators in TDL channel model.}
    \label{fig: mse_tdl}
  \end{subfigure}
  \hfill 
  \begin{subfigure}[b]{0.30\textwidth}
    \centering
    \includegraphics[width=\linewidth]{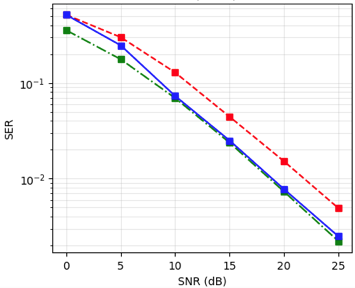}
    \caption{SER comparison for the estimators in TDL channel model.}
    \label{fig: ser_tdl}
  \end{subfigure}
  \hfill
  \begin{subfigure}[b]{0.30\textwidth}
    \centering
    \includegraphics[width=\linewidth]{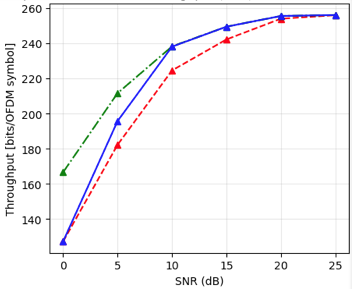}
    \caption{THROUGHPUT comparison for the estimators in TDL channel model.}
    \label{fig:throughput_tdl}
  \end{subfigure}
  \caption{Algorithm performance in TDL channel model.}
  \label{fig:TDL}
\end{figure}

Right away, we analyzed the performance of the algorithms in the ITU channel model.
Fig.~\ref{fig:ITU} shows that the subgroup-based estimator, even in the worst-case scenario, continues outperforms the LS methods for all SNR range. Despite the MSE remaining distant to the LMMSE method, it's SER and throughput get closer to the LMMSE method when the SNR exceeds 15 dB. This test shows that, even in the worst-case scenario, the proposed method performs well with low computational complexity.


\begin{figure}[h!] 
  \centering 
  \begin{subfigure}[b]{0.30\textwidth} 
    \centering
    \includegraphics[width=\linewidth]{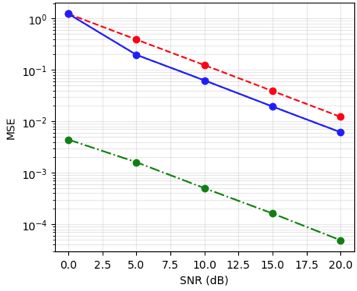}
    \caption{MSE comparison for the estimators in ITU channel model.}
    \label{fig: mse_itu}
  \end{subfigure}
  \hfill 
  \begin{subfigure}[b]{0.30\textwidth}
    \centering
    \includegraphics[width=\linewidth]{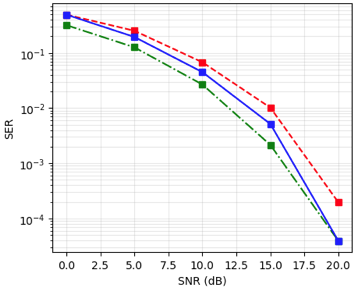}
    \caption{SER comparison for the estimators in ITU channel model.}
    \label{fig: ser_itu}
  \end{subfigure}
  \hfill
  \begin{subfigure}[b]{0.30\textwidth}
    \centering
    \includegraphics[width=\linewidth]{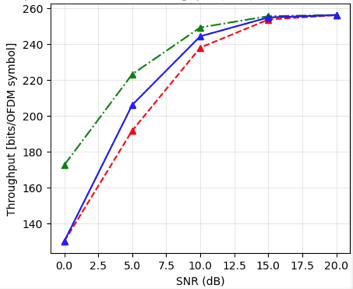}
    \caption{THROUGHPUT comparison for the estimators in ITU channel model.}
    \label{fig:throughput_itu}
  \end{subfigure}
  \caption{Algorithm performance in ITU channel model.}
  \label{fig:ITU}
\end{figure}

\section{Conclusion}
\label{sec:conclusion}

This work introduced a group-theoretic framework for structured channel estimation in OFDM systems.  
By representing subcarriers as the cyclic group \(\mathbb{Z}_N\) and exploiting the annihilator relation between subgroups \(H\) and \(H^\perp\),  
we established a direct correspondence between frequency-domain nulling and time-domain sparsity.  
The proposed estimator leverages this structure through a simple energy-based criterion, achieving substantial gains in mean squared error, bit error rate, and throughput compared with classical and data-driven baselines.  
These results confirm that algebraic regularities can be systematically harnessed to improve estimation accuracy and spectral efficiency without increasing model complexity.  
Future work will extend the framework to non-abelian group structures and multi-antenna scenarios, where richer symmetries may further enhance performance in dynamic wireless environments.  
Overall, the results suggest that symmetry-aware estimators provide a promising direction for next-generation wireless systems.

\begin{appendices}

\end{appendices}


\bibliography{sn-bibliography}

@ARTICLE{DiRenzo2020,
  author={Di Renzo, Marco and Zappone, Alessio and Debbah, Merouane and Alouini, Mohamed-Slim and Yuen, Chau and de Rosny, Julien and Tretyakov, Sergei},
  journal={IEEE Journal on Selected Areas in Communications}, 
  title={Smart Radio Environments Empowered by Reconfigurable Intelligent Surfaces: How It Works, State of Research, and The Road Ahead}, 
  year={2020},
  volume={38},
  number={11},
  pages={2450-2525},
  doi={10.1109/JSAC.2020.3007211}
}

@ARTICLE{Coleri2002,
  author={Coleri, S. and Ergen, M. and Puri, A. and Bahai, A.},
  journal={IEEE Transactions on Broadcasting}, 
  title={Channel estimation techniques based on pilot arrangement in OFDM systems}, 
  year={2002},
  volume={48},
  number={3},
  pages={223-229},
  doi={10.1109/TBC.2002.804034}
}

@BOOK{LiStuber2006,
  editor       = {Ye Li and Gordon L. St\"uber},
  title        = {Orthogonal Frequency Division Multiplexing for Wireless Communications},
  series       = {Signals and Communication Technology},
  publisher    = {Springer},
  address      = {New York, NY},
  year         = {2006},
  edition      = {1},
  pages        = {XII, 308},
  doi          = {10.1007/0-387-30235-2},
  isbn         = {978-0-387-29095-9}
}

@book{Goldsmith_2005,
  author       = {Goldsmith, Andrea},
  title        = {Wireless Communications},
  publisher    = {Cambridge University Press},
  address      = {Cambridge},
  year         = {2005},
  isbn         = {9780511841224},
  doi          = {10.1017/CBO9780511841224}
}

@ARTICLE{Barhumi2003,
  author={Barhumi, I. and Leus, G. and Moonen, M.},
  journal={IEEE Transactions on Signal Processing}, 
  title={Optimal training design for MIMO OFDM systems in mobile wireless channels}, 
  year={2003},
  volume={51},
  number={6},
  pages={1615-1624},
  doi={10.1109/TSP.2003.811243}
}

@book{HeathLozano2018,
  author    = {Robert W. Heath Jr. and Angel Lozano},
  title     = {Foundations of MIMO Communication},
  publisher = {Cambridge University Press},
  year      = {2018},
  address   = {Cambridge},
  isbn      = {9781139049276},
  doi       = {10.1017/9781139049276},
  note      = {Online publication December 2018}
}

@book{VanTrees2002,
  author    = {Harry L. Van Trees},
  title     = {Optimum Array Processing: Part IV of Detection, Estimation, and Modulation Theory},
  publisher = {John Wiley \& Sons},
  year      = {2002},
  address   = {New York},
  isbn      = {9780471093909},
  doi       = {10.1002/0471221104}
}

@ARTICLE{Basar2019,
  author={Basar, Ertugrul and Di Renzo, Marco and De Rosny, Julien and Debbah, Merouane and Alouini, Mohamed-Slim and Zhang, Rui},
  journal={IEEE Access}, 
  title={Wireless Communications Through Reconfigurable Intelligent Surfaces}, 
  year={2019},
  volume={7},
  pages={116753-116773},
  doi={10.1109/ACCESS.2019.2935192}
}

@ARTICLE{Bajwa2010,
  author={Bajwa, Waheed U. and Haupt, Jarvis and Sayeed, Akbar M. and Nowak, Robert},
  title={Compressed Channel Sensing: A New Approach to Estimating Sparse Multipath Channels},
  journal={Proceedings of the IEEE},
  year={2010},
  volume={98},
  number={6},
  pages={1058-1076},
  doi={10.1109/JPROC.2010.2042415}
}

@book{Cho2010,
  author    = {Yong Soo Cho and Jaekwon Kim and Won Young Yang and Chung G. Kang},
  title     = {MIMO-OFDM Wireless Communications with MATLAB},
  publisher = {John Wiley \& Sons},
  year      = {2010},
  address   = {Hoboken, NJ},
  isbn      = {9780470825624}
}

@ARTICLE{Ai2014,
  author={Ai, Bo and Cheng, Xiang and Kürner, Thomas and Zhong, Zhang-Dui and Guan, Ke and He, Rui-Si and Xiong, Lei and Matolak, David W. and Michelson, David G. and Briso-Rodriguez, Cesar},
  journal={IEEE Transactions on Intelligent Transportation Systems}, 
  title={Challenges toward wireless communications for high-speed railway}, 
  year={2014},
  volume={15},
  number={5},
  pages={2143-2158},
  doi={10.1109/TITS.2014.2310771}
}

@ARTICLE{Chafii2023,
  author={Chafii, Marwa and Bariah, Lina and Muhaidat, Sami and Debbah, Merouane},
  journal={IEEE Communications Surveys \& Tutorials}, 
  title={Twelve Scientific Challenges for 6G: Rethinking the Foundations of Communications Theory}, 
  year={2023},
  volume={25},
  number={2},
  pages={868-904},
  doi={10.1109/COMST.2022.3224532}
}

@INPROCEEDINGS{Tong2022,
  author={Tong, Wen and Zhu, Peiying},
  booktitle={2022 IEEE 96th Vehicular Technology Conference (VTC2022-Fall)}, 
  title={6G: The Next Horizon}, 
  year={2022},
  pages={1-5},
  doi={10.1109/VTC2022-Fall57202.2022.10012887}
}

@INPROCEEDINGS{Hadani2017,
  author={Hadani, Ron and Rakib, Shlomo and Tsatsanis, Michail and Monk, Andy and Goldsmith, Andrea J. and Molisch, Andreas F. and Calderbank, Robert},
  booktitle={2017 IEEE Wireless Communications and Networking Conference (WCNC)}, 
  title={Orthogonal Time Frequency Space Modulation}, 
  year={2017},
  pages={1-6},
  doi={10.1109/WCNC.2017.7925924}
}

@ARTICLE{Ye2017,
  author={Ye, Hao and Li, Geoffrey Ye and Juang, Borching},
  journal={IEEE Wireless Communications Letters}, 
  title={Power of Deep Learning for Channel Estimation and Signal Detection in OFDM Systems}, 
  year={2017},
  volume={7},
  number={1},
  pages={114-117},
  doi={10.1109/LWC.2017.2757490}
}

@ARTICLE{Honkala2021,
  author={Honkala, Mikko and Korpi, Dani and Huttunen, Janne M. J.},
  journal={IEEE Transactions on Wireless Communications}, 
  title={DeepRx: Fully Convolutional Deep Learning Receiver}, 
  year={2021},
  volume={20},
  number={6},
  pages={3925-3940},
  doi={10.1109/TWC.2021.3054520}
}

@ARTICLE{Aoudia2021,
  author={Aoudia, Fayçal Ait and Hoydis, Jakob},
  journal={IEEE Transactions on Wireless Communications}, 
  title={End-to-End Learning for OFDM: From Neural Receivers to Pilotless Communication}, 
  year={2021},
  volume={21},
  number={2},
  pages={1049-1063},
  doi={10.1109/TWC.2021.3101306}
}

@book{PapoulisPillai2002,
  author    = {Athanasios Papoulis and S. Unnikrishna Pillai},
  title     = {Probability, Random Variables, and Stochastic Processes},
  publisher = {McGraw-Hill},
  year      = {2002},
  address   = {New York},
  series    = {McGraw-Hill Series in Electrical and Computer Engineering},
  edition   = {4},
  isbn      = {9780073660110}
}

@book{Sayed2003,
  author    = {Sayed, Ali H.},
  title     = {Fundamentals of Adaptive Filtering},
  publisher = {John Wiley \& Sons},
  year      = {2003},
  address   = {Hoboken, NJ}
}

@article{Molisch2012,
  title={Modeling the frequency dependence of ultra-wideband spatio-temporal indoor radio channels},
  author={Haneda, Katsuyuki and Richter, Andreas and Molisch, Andreas F},
  journal={IEEE transactions on antennas and propagation},
  volume={60},
  number={6},
  pages={2940--2950},
  year={2012},
  publisher={IEEE}
}

@book{durgin2003,
  author    = {Gregory D. Durgin},
  title     = {Space-time wireless channels},
  publisher = {Prentice Hall},
  address   = {Upper Saddle River, NJ},
  year      = {2003},
  isbn      = {013065647X}
}

@article{Rusek2013,
  title={Scaling up MIMO: Opportunities and challenges with very large arrays, IEEE Signal Process},
  author={Rusek, F and Persson, D and Lau, BK and Larsson, EG and Marzetta, TL and Edfors, O and Tufvesson, F},
  journal={Mag},
  volume={30},
  number={1},
  pages={40--60},
  year={2013}
}

@inproceedings{Tang2021,
  title={Use of intelligent reflecting surfaces for and against wireless communication security},
  author={Sarp, Salih and Tang, Haolin and Zhao, Yanxiao},
  booktitle={2021 IEEE 4th 5G World Forum (5GWF)},
  pages={374--377},
  year={2021},
  organization={IEEE}
}

@article{Zhou2020,
  title={Channel modeling and characteristics for 6G wireless communications},
  author={Jiang, Hao and Mukherjee, Mithun and Zhou, Jie and Lloret, Jaime},
  journal={IEEE Network},
  volume={35},
  number={1},
  pages={296--303},
  year={2020},
  publisher={IEEE}
}

@article{Wu2022,
  title={Intelligent reflecting surface enabled multi-target sensing},
  author={Meng, Kaitao and Wu, Qingqing and Schober, Robert and Chen, Wen},
  journal={IEEE Transactions on Communications},
  volume={70},
  number={12},
  pages={8313--8330},
  year={2022},
  publisher={IEEE}
}

@article{Stojanovic2008,
  title={Multicarrier communication over underwater acoustic channels with nonuniform Doppler shifts},
  author={Li, Baosheng and Zhou, Shengli and Stojanovic, Milica and Freitag, Lee and Willett, Peter},
  journal={IEEE Journal of Oceanic Engineering},
  volume={33},
  number={2},
  pages={198--209},
  year={2008},
  publisher={IEEE}
}

@techreport{ITU-R:M.1225,
  author = {{ITU-R}},
  title = {{Guidelines for evaluation of radio transmission technologies for IMT-2000}},
  institution = {{International Telecommunication Union}},
  year = {1997},
  number = {M.1225},
  url = {https://www.itu.int/rec/R-REC-M.1225-0-199702-I/en}
}

@ARTICLE{9305216,
  author={Qiu, Hao and García-Loygorri, Juan Moreno and Guan, Ke and He, Danping and Xu, Ziheng and Ai, Bo and Berbineau, Marion},
  journal={IEEE Access}, 
  title={Emulation of Radio Technologies for Railways: A Tapped-Delay-Line Channel Model for Tunnels}, 
  year={2021},
  volume={9},
  number={},
  pages={1512-1523},
  doi={10.1109/ACCESS.2020.3046852}}

@ARTICLE{701321,
  author={Edfors, O. and Sandell, M. and van de Beek, J.-J. and Wilson, S.K. and Borjesson, P.O.},
  journal={IEEE Transactions on Communications}, 
  title={OFDM channel estimation by singular value decomposition}, 
  year={1998},
  volume={46},
  number={7},
  pages={931-939},
  doi={10.1109/26.701321}}

\end{document}